\documentclass[reprint, amsmath,amssymb, aps]{revtex4-1}
\usepackage{amsthm}
\newtheorem{theorem}{Theorem}
\newtheorem{lemma}{Lemma}
\theoremstyle{plain}
\usepackage{bm}
\usepackage{graphicx}
\usepackage{hyperref}
\usepackage{physics}
\usepackage{enumerate}
\begin{document}
%\preprint{APS/123-QED}
\date{\today}
\author{Nikolaos K. Kollas}
\email{kollas@upatras.gr}
\affiliation{Department of Physics, University of Patras, Greece}
\title{Optimization free measures of quantum resources}
\begin{abstract}
A closed expression is derived for the amount of resource present in a quantum state as quantified by a distance measure based on the \emph{Tsallis relative entropy} introduced recently in \cite{Zhao2018}, for any resource theory whose set of free states is described by the image of a unital \emph{resource destroying map}. By applying it to the resource theory of coherence it is demonstrated that the correct definition of a \emph{projective coarse grained measurement} needs to be modified in order for it to be correctly interpreted as a measurement which provides less information about a system's true state than the one obtained from a \emph{finer grained} measurement.
\end{abstract}
\maketitle
\section{Introduction}
A \emph{resource} is defined as any physical object that is consumed during a process in order to perform some kind of useful action, such as the burning of fossil fuel for the generation of mechanical work to operate machinery or the amount of processing power needed to perform a computation. In the context of quantum information theory, a \emph{quantum resource} is similarly defined as any state or channel which can be utilized in order to simulate quantum operations which, due to certain restrictions, are otherwise unavailable to us. A well known example includes the use of an entangled Bell state needed for teleporting an unknown quantum state between two parties \cite{Bennett1993}. Over the past few years a number of researchers have developed various quantum resource theories such as the resource theory of entanglement \cite{Horodecki2009}, purity \cite{Horodecki2003a}, coherence \cite{Baumgratz2014,Chitambar2016,Winter2016,Marvian2016,Napoli}, athermality \cite{Brandao2013,Horodecki2013a,Gour2015} and asymmetry \cite{Gour2008,Gour2009,Marvian2014} amongst others.  In any formulation one has to distinguish between the states that constitute a resource and those that are resource free and then construct a suitable measure for the amount of resource present in each case.

One of the most intuitive and straightforward ways to do this involves determining how close the given state is to the set of resource free states denoted by $\mathcal{F}$. This can be achieved by introducing a distance function on the set of density matrices, the amount of resource that a state possesses will then be equal to it's distance from $\mathcal{F}$. The logic behind such a distance based measure is simple, the further away the state lies from the set of resource free states the more useful it must be as a resource. Unfortunately for most distance functions the minimization required can only be obtained by optimization methods making their use impractical.

In this letter we extend a result by Zhao et. al. \cite{Zhao2018} and prove that by restricting attention to resource theories whose set of resource free states is equal to the fixed points of an \emph{idempotent} and \emph{unital} resource destroying map one can construct a continuous family of resource measures given by a closed expression. After a short discussion on the formulation of quantum resource theories we shall present this closed form and discuss it's properties. We then apply iy to the resource theory of coherence and demonstrate that the notion of a \emph{coarse grained measurement} requires modification for it to be consistently interpreted as a measurement which yields less information about a system's true state than a \emph{finer grained} one.

In what follows we shall focus on finite dimensional systems with Hilbert space $\mathbb{C}^d$, quantum states are then equal to the set of positive definite Hermitian operators with unit trace. 
\section{Quantum resource theories}\label{sec1}
To construct a theory of \emph{quantum resources} one begins by specifying the set of \emph{resource free} states $\mathcal{F}$. The class of allowed operations $\mathcal{L}$ that we are allowed to implement can then be defined as those trace preserving and completely positive (CPTP) quantum operations $\Lambda$ which leave the set invariant.
\begin{equation*}
	\mathcal{L}:=\qty{\Lambda\in(CPTP)|\Lambda(\sigma)\in\mathcal{F},\forall \sigma \in \mathcal{F}}
\end{equation*}
By definition any state $\rho\not\in\mathcal{F}$ is considered to be a \emph{resource} and can be used to implement operations outside of $\mathcal{L}$ temporarily overcoming any restrictions that were imposed.

By employing the operator sum representation of a quantum operation we can define an even stronger set of allowed operations, namely let $\Lambda({\sigma})=\sum_i\Lambda_i\sigma\Lambda_i^\dagger$ with $\sum_i\Lambda_i^\dagger\Lambda_i=I$ then
\begin{equation*}
	\bar{\mathcal{L}}:=\qty{\Lambda\in(CPTP)\bigg| \frac{\Lambda_i\sigma\Lambda_i^\dagger}{p_i}\in\mathcal{F},\forall \sigma \in \mathcal{F},\forall i}
\end{equation*}
where $p_i=Tr\Lambda_i\sigma\lambda_i^\dagger$. Clearly $\bar{\mathcal{L}}\subseteq\mathcal{L}$. The physical motivation behind the stronger subset of allowed operations is that retaining a record of the classical variable $i$, a process which is physically realizable in the lab \cite{Baumgratz2014},  will not produce a resource state.

A few examples include the theory of \emph{entanglement} \cite{Horodecki2009} where the class of \emph{Local Operations and Classical Communication} (LOCC) leaves the set of \emph{separable states} invariant, the resource theory of \emph{coherence} and \emph{purity} where \emph{Incoherent} (IO) and unitary operations leave invariant a set of states which are diagonal with respect to a fixed basis and the totally mixed state respectively, \cite{HorodeckiOppenheim2013} and the theory of \emph{quantum reference frames} and \emph{asymmetry} where a $G$-covariant operation leaves the set of symmetric states invariant, $G$ being the group describing the symmetry in question. \cite{Bartlett2007,Gour2008,Gour2009}. In some cases it is also possible to consider two resources jointly, this seems to be the case for entanglement and coherence under the extended class of Local Incoherent Operations with Classical Communication (LIOCC) \cite{Streltsov2015,Chitambar2016}.

Once a resource theory has been formulated in such terms it is necessary to find a way of measuring the amount of resource present in any given state. This is made possible with the help of suitable \emph{resource measures}, i.e. non-negative real valued functions on the set of density matrices. In order to qualify as a resource measure any candidate $\mu$ is required to have the following properties
\begin{enumerate}[i)]
	\item $\mathcal{F}\subseteq ker(\mu)$ i.e $\mu(\sigma)=0$, $\forall\sigma\in\mathcal{F}$
    \item $\mu(U\rho U^\dagger)=\mu(\rho)$ for every unitary operation $U\in\bar{\mathcal{L}}$
	\item $\mu$ must be non-increasing on average, specifically if $\Lambda(\rho)=\sum_ip_i\rho_i$ then $\sum_ip_i\mu(\rho_i)\leq\mu(\rho)$ for every $\Lambda\in\bar{\mathcal{L}}$
    \item $\mu$ must be non-increasing under mixing i.e. $\mu(\sum_ip_i\rho_i)\leq\sum_ip_i\mu(\rho_i)$
\end{enumerate}
By combining properties iii) and iv) which are known as the \emph{strong monotonicity} and \emph{convexity} properties respectively one can easily show that $\mu(\Lambda(\rho))\leq\mu(\rho)$, $\forall\Lambda\in\bar{\mathcal{L}}$.

One way of constructing such measures involves considering the distance of the state from the set of free states $\mathcal{F}$. Specifically any distance function $d(\rho,\sigma)$ on the set of density matrices induces a corresponding \emph{distance measure} $\mu_d(\rho)$ defined as
\begin{equation}\label{eq1}
	\mu_d(\rho)=\min_{\sigma\in\mathcal{F}}d(\rho,\sigma)
\end{equation}
Examples of some of the most commonly used distances include the \emph{relative entropy} \cite{Vedral1998} $S(\rho|\sigma)=-\tr\rho\log\sigma+\tr\rho\log\rho$, the \emph{trace norm} $\norm{\rho-\sigma}_1=\tr\abs{\rho-\sigma}$ and the \emph{Hilbert-Schmidt} norm $\norm{\rho-\sigma}_2=\sqrt{\tr(\rho-\sigma)^2}$ \cite{bhatia1997}, though technically speaking the relative entropy does not qualify as a distance function since it is not symmetric in it's arguments ($S(\rho|\sigma)\neq S(\sigma|\rho)$) it is frequently used because of it's operational interpretation and it's connection to asymptotic conversion rates \cite{HorodeckiOppenheim2013,Brandao2015}.

\section{Optimization free measures}
We now restrict our attention to those resource theories whose set of free states is equal to the image of an idempotent and unital \emph{resource destroying map} \cite{Liu2017,Gour2017}. Namely let $\mathcal{E}$ be a completely positive and trace preserving linear operation with the following properties
\begin{itemize}
	\item $\mathcal{E}^2=\mathcal{E}$ (idempotency).
    \item $\mathcal{E}(I)=I$ (unitality).
\end{itemize}

With the exception of entanglement for all of the resource theories referred to earlier the set of free states is given by the image of such a suitably chosen operation on the set of density matrices i.e. $\mathcal{F}\equiv Im(\mathcal{E})$. This is indeed the case for the resource theory of coherence where $\mathcal{E}(\rho)=\sum_iP_i\rho P_i$ is the state after a projective measurement is performed in a fixed basis, the resource theory of purity where by definition $\mathcal{E}(\rho)=\frac{I}{d}$ and $d$ is the dimension of the underlying Hilbert space \footnote{Note that this is equivalent to the resource theory of assymmetry in the case of a degenerate Hamiltonian.}, and the resource theory of quantum reference frames and asymmetry where $\mathcal{E}(\rho)=\int U(g)\rho U^\dagger(g)dg$ is the G-twirling operation associated with the group G, $U(g)$ is the unitary representation of the group and $dg$ the invariant measure. Noting that due to idempotency the image of $\mathcal{E}$ is equal to it's set of fixed points, i.e. $Fix(\mathcal{E})=\{\sigma|E(\sigma)=\sigma\}$, Equation (\ref{eq1}) can then be rewritten as
\begin{equation*}
	\mu_d^{\mathcal{E}}(\rho)=\min_{\sigma\in Fix(\mathcal{E})}d(\rho,\sigma)
\end{equation*}
We now employ the \emph{Tsallis based relative entropy distance} introduced in\cite{Zhao2018}
\begin{equation*}
	\tilde{S_a}(\rho|\sigma)=
   	\begin{cases}
    	\frac{(Tr\rho^a\sigma^{1-a})^{\frac{1}{a}}-1}{a-1} & a\in(0,1)\cup(1,2]\\
        Tr\rho\log\rho-Tr\rho\log\sigma & a=1
    \end{cases}
\end{equation*}
In this case we generalize  their basic result about the coherence of a quantum state and show that
\begin{theorem}
		$$\min_{\sigma\in Fix(\mathcal{E})}\tilde{S}_a(\rho|\sigma)=\frac{Tr\mathcal{E}^{\frac{1}{a}}(\rho^a)-1}{a-1}$$
\end{theorem}
\begin{proof}
	To begin with we note that as was already mentioned in \cite{Gour2009} (where the proof for $a=1$ was given) if f(x) is analytic at x=1 then $f(\sigma)\in\ Fix(\mathcal{E})$, $\forall\sigma\in Fix(\mathcal{E})$. Also for any unital superoperator one can define it's adjoint as $tr\mathcal{E}^\dagger(\rho)\sigma=tr\rho\mathcal{E}(\sigma)$ and further show that $Fix(\mathcal{E}^\dagger)=Fix(\mathcal{E})$.  Applying this to the function $f(x)=x^{\frac{1}{a}}$ for $a\in(0,1)\cup(1,2]$ and setting $N=Tr\mathcal{E}^{\frac{1}{a}}(\rho^a)$ we find that
\begin{equation*}
	\begin{split}
    	Tr\rho^a\sigma^{1-a}&=Tr\rho^a\mathcal{E}^\dagger(\sigma^{1-a})\\
        &=Tr\mathcal{E}(\rho^a)\sigma^{1-a}\\
        &=N^aTr\left(\frac{\mathcal{E}^{\frac{1}{a}}(\rho^a)}{N}\right)^a\sigma^{1-a}
    \end{split}
\end{equation*}%
So
\begin{equation*}
    	\tilde{S}_a(\rho|\sigma)=\frac{N-1}{a-1}+\tilde{S}_a\left(\frac{\mathcal{E}^{1/a}(\rho^a)}{N}\bigg|\sigma\right)
\end{equation*}
and the minimum is given by $\sigma=\frac{\mathcal{E}^{\frac{1}{a}}(\rho^a)}{N}$.
\end{proof}

Due to it's definition the Tsallis based relative entropy satisfies properties i) to iii) \cite{Vedral1998,Zhao2018}. We now proceed to show that it is also convex.
\begin{lemma}
$\mu_{\tilde{S}_a}^\mathcal{E}(\rho)$ is operator convex.
\end{lemma}
\begin{proof}
	It is known that for any positive matrix $A$, $tr(XA^aX^\dagger)^\frac{1}{a}$ is concave for $ a\in(0,1]$ and convex for $a\in[1,2]$ \cite{Carlen2008}. Let $\mathcal{E}(\rho)=\sum_{i=1}^nE_i\rho E^\dagger_i$ be the operator sum representation of $\mathcal{E}$ and let
\begin{equation*}
X=\begin{pmatrix}
	E_1 & E_2 &\cdots& E_n\\
    0&0&\cdots & 0\\
    \vdots&\vdots&\vdots&\vdots\\
    0&0&\cdots&0\end{pmatrix}
 \quad A=\begin{pmatrix}
 	\rho&0&\cdots&0\\
    0&\rho&\cdots&0\\
   \vdots&\vdots&\vdots&\vdots\\
   0&0&\cdots&\rho
 \end{pmatrix}
\end{equation*}
where $X$ and $A$ are matrices in the direct sum of $n$ copies of $\mathcal{H}$, $\mathcal{H}'=\bigoplus_{i=1}^n\mathcal{H}$. It follows that
\begin{equation*}
XA^aX^\dagger=
	\begin{pmatrix}
		\mathcal{E}(\rho^a)&0&\cdots&0\\
    0&0&\cdots&0\\
   \vdots&\vdots&\vdots&\vdots\\
   0&0&\cdots&0
	\end{pmatrix}
\end{equation*}
and $Tr_{\mathcal{H}'}(XA^aX^\dagger)^\frac{1}{a}=Tr_{\mathcal{H}}\mathcal{E}^\frac{1}{a}(\rho^a)$. Combining this with the definition of the Tsallis based relative entropy completes the proof.
\end{proof}
\section{Resource theory of projective quantum measurements.}
The act of measurement consists of gaining information about a system by e.g. measuring the system's position or momentum, and then using this information in order to make an estimate of the system's true state \cite{Marehand1977}. It is natural to assume that by repeating the same measurement we obtain no new information and that no more information can be extracted from a completely mixed state. Quantum mechanically both of these assumptions are satisfied by a \emph{projective quantum measurement}. This is described by a set of mutually orthogonal projectors $P_i$ summing to the identity. The state obtained after outcome $i$ has occurred with probability $p_i=Tr\rho P_i$ is then equal to $\rho_i=\frac{P_i\rho P_i}{p_i}$ and our estimate is given by $\Pi(\rho)=\sum_i P_i\rho P_i$ (in what follows we shall identify projective measurements by their corresponding set of projections and write $\Pi=\{P_i\}$). As was already mentioned by keeping the set of projectors fixed we obtain the resource theory of coherence. 
\subsection{Coarse grained and fine grained measurements.}
Let $n_i=Tr P_i$ be the degeneracy of the i-th projection. In a \emph{fine grained} (or Von Neumann \cite{von1955mathematical}) measurement $n_i=1$, $\forall i$ while in a \emph{coarse grained} (or L\"uders \cite{ANDP:ANDP200610207}) measurement $n_i\geq 1$. We now consider a special class of coarse grained measurements where each projection $L_i$ is a sum of fine grained projections. Specifically let $\Pi=\{P_i\}$, $i\in I$ be a fixed fine grained measurement then the corresponding coarse grained measurement $\bar{\Pi}=\{L_j\}$, $j\in J$ is obtained by partitioning the set of indexes $i$ into $J$ separate subsets $I_j$ and setting $L_j=\sum_iP_i\chi_{I_j}(i)$, $j\in J$ where $\chi_{I_j}$ is the characteristic function of set $I_j$. 
In \cite{Piani2014} the authors pointed out that for any distance $d$, $d(\rho,\Pi(\rho))\geq d(\rho,\bar{\Pi}(\rho))$, this is counter-intuitive as we would expect that by performing fine grained measurements we gain more information about the state and that our estimate would get closer to it. Following \cite{Wehrl1977} let us slightly modify our estimate of the state after the measurement and demand that it be equal to $\tilde{\Pi}(\rho)=\sum_ip_i\frac{L_i}{Tr L_i}$ where $p_i=Tr\rho L_i$, note that if each $L_i$ is one dimensional then $\tilde\Pi$ is equal to the fine-grained measurement. We now prove the following.
\begin{theorem}
 $$\mu_{\tilde{S}_a}^\Pi(\rho)\leq\mu_{\tilde{S}_a}^{\tilde{\Pi}}(\rho)$$
\end{theorem}
\begin{proof}
	Since $\Pi^2=\Pi$, $\tilde\Pi^2=\tilde\Pi$ and $\Pi(I)=\tilde\Pi(I)=I$ according to Theorem 1
$$\mu_{\tilde{S}_a}^\Pi(\rho)=\min_{\sigma\in Fix(\Pi)}\tilde{S}_a(\rho|\sigma)$$
and 
$$\mu_{\tilde{S}_a}^{\tilde{\Pi}}(\rho)=\min_{\sigma\in Fix(\tilde \Pi)}\tilde{S}_a(\rho|\sigma)$$
It is easy now to see that $\Pi(\tilde{\Pi}(\rho))=\tilde{\Pi}(\rho)$ while $\tilde\Pi(\Pi(\rho))=\tilde{\Pi}(\rho)$, this means that $Im(\tilde\Pi)\subset Im(\Pi)$ from which the theorem immediately follows.
\end{proof}
\section{Discussion and conclusions.}
We have shown that for any quantum resource theory whose set of free states is given by the image of an idempotent and unital resource destroying map $\mathcal{E}$ on the set of density matrices, the amount of resource present in a state is given by a closed form when the Tsallis based relative entropy is used as a distance based measure.  By applying it to the resource theory of projective measurements we demonstrated how to properly define coarse grained measurements such that the information obtained of a system's true state is increased under a finer grained measurement.

In future research \cite{Kollas} it will be shown how for bipartite systems, taking into account the set of allowed projective measurements that each observer can implement on it's respective subsystem, and by employing an optimization procedure involving the set of allowed unitary operations  we can recover measures of quantum correlations such as the \emph{quantum discord} which is known to be equal to the \emph{entropy of entanglement} in the special case where the state is pure \cite{Bravyi2003}. The same approach can also be applied to systems which do not admit a subsystem decomposition such as a qutrit, revealing a new novel feature of quantum mechanics, namely the restrictions on acquiring information about a state imposed by symmetry considerations. 

Finally an interesting open question concerns the case when the set of free states is described by a resource destroying map which is not unital, as is the case in the resource theory of athermality and the recently developed resource theory of imaginarity \cite{Hickey2018}.

\section*{Acknowledgments}
The author wishes to thank Charis Anastopoulos, Konstantinos Blekos, and D. Kat and for many fruitful discussions which helped in presenting the ideas in this paper. This research was supported by Grant No. E611 from the Research Committee of the University of Patras via the ''K. Karatheodoris'' program.
\bibliographystyle{apsrev4-1}
\bibliography{supplementaryV2}
\end{document}